\def\final{1}

\documentclass[11pt]{article}
	
\usepackage{amsthm, amsmath, amssymb}
\usepackage{enumerate}
\usepackage{enumitem}
\usepackage{framed}
\usepackage{verbatim}
\usepackage[left=1.1in, right=1.1in, top=1.0in, bottom=1.0in]{geometry}
\usepackage{microtype}
\usepackage{kpfonts}
	\DeclareMathAlphabet{\mathsf}{OT1}{cmss}{m}{n}
	\SetMathAlphabet{\mathsf}{bold}{OT1}{cmss}{bx}{n}
	\DeclareMathAlphabet{\mathtt}{OT1}{cmtt}{m}{n}
	\SetMathAlphabet{\mathtt}{bold}{OT1}{cmtt}{bx}{n}

\usepackage{multirow, tabularx}
\usepackage{algorithm, algorithmic}
\usepackage{color}
	\definecolor{DarkGreen}{rgb}{0.15,0.5,0.15}
	\definecolor{DarkRed}{rgb}{0.6,0.2,0.2}
	\definecolor{DarkBlue}{rgb}{0.15,0.15,0.55}
	\definecolor{DarkPurple}{rgb}{0.4,0.2,0.4}

\usepackage[pdftex]{hyperref}
\hypersetup{
	linktocpage=true,
	colorlinks=true,			
	linkcolor=DarkBlue,	
	citecolor=DarkBlue,	
	urlcolor=DarkBlue,		
	}
	
\newcolumntype{Y}{>{\centering\arraybackslash}X}

\ifnum\final=0
\newcommand{\mynote}[2]{{\color{#1} \marginpar{\tiny #2}}}
\newcommand{\mybignote}[2]{{\color{#1} $\langle \langle$ #2$\rangle \rangle$}}
\else
\newcommand{\mynote}[2]{}
\newcommand{\mybignote}[2]{}
\fi
\newcommand{\jnote}[1]{\mynote{DarkRed}{Jon: {#1}}}

\newcommand{\tnote}[1]{\mynote{magenta}{Thomas: {#1}}}

\newcommand{\pr}[2]{\underset{#1}{\mathbb{P}}\left[ #2 \right]}
\newcommand{\prob}[1]{\mathbb{P}\left[ #1 \right]}
\newcommand{\ex}[2]{\underset{#1}{\mathbb{E}}\left[ #2 \right]}
\newcommand{\expe}[1]{\mathbb{E}\left[ #1 \right]}
\newcommand{\var}[2]{\underset{#1}{\mathrm{Var}}\left[ #2 \right]}

\newcommand{\eps}{\varepsilon}

\DeclareMathOperator*{\argmax}{arg\,max}

\newcommand{\E}{\mathbb{E}}
\newcommand{\N}{\mathbb{N}}
\newcommand{\R}{\mathbb{R}}

\newcommand{\cS}{\mathcal{S}}

\newtheorem{theorem}{Theorem}[section]

\newtheorem{lemma}[theorem]{Lemma}
\newtheorem{lem}[theorem]{Lemma}

\newtheorem{claim}[theorem]{Claim}

\newtheorem{prop}[theorem]{Proposition}

\theoremstyle{definition}

\newtheorem{definition}[theorem]{Definition}

\newcommand{\stab}{\eta}
\newcommand{\sel}{\cS}
\newcommand{\matx}{\mathbf{x}}
\newcommand{\vecx}{\tilde{\mathbf{x}}}
\newcommand{\matX}{\mathbf{X}}
\renewcommand{\t}{j}
\newcommand{\T}{m}

\title{Subgaussian Tail Bounds via Stability Arguments}
\author{Thomas Steinke\thanks{IBM, Almaden Research Center. \texttt{dp-tail@thomas-steinke.net}} \and Jonathan Ullman\thanks{Northeastern University. \texttt{jullman@ccs.neu.edu}}}

\begin{document}

\maketitle


\begin{abstract}
Sums of independent, bounded random variables concentrate around their expectation approximately as well a Gaussian of the same variance.  Well known results of this form include the Bernstein, Hoeffding, and Chernoff inequalities and many others.  We present an alternative proof of these tail bounds based on what we call a \emph{stability argument,} which avoids bounding the moment generating function or higher-order moments of the distribution.  Our stability argument is inspired by recent work on the generalization properties of differential privacy and their connection to adaptive data analysis (Bassily et al., STOC 2016).
\end{abstract}



\newcommand{\e}[1]{e^{#1}}

\section{Introduction}
Tail bounds, also known as large deviation inequalities, or concentration of measure theorems, are an essential tool in mathematics, computer science, and statistics.  The most common type of tail bounds (in theoretical computer science) show that the sum of independent, bounded random variables has subgaussian tails---the best known example being the following special case of Bernstein's Inequality.
\begin{theorem}[\cite{Bernstein24}]\label{thm:Hoeffding-intro} If $X_1, \cdots, X_n$ are independent random variables supported on $[0,1]$ and $\mu_i=\expe{X_i}$ for every $i$, then
\begin{equation*}
\forall \varepsilon \geq 0 ~~~~~ \prob{\sum_{i=1}^n X_i -\mu_i \geq \varepsilon n} \leq e^{-\Omega(\varepsilon^2n)}.\label{-eqn:Hoeffding-intro}
\end{equation*}
\end{theorem}

Theorem \ref{thm:Hoeffding-intro} has many generalizations and extensions: tight constants~\cite{Hoeffding63}, random variables with differing ranges~\cite{Hoeffding63}, multiplicative error bounds~\cite{Chernoff52}, variables with low variance \cite{Bernstein24}, variables with only limited independence \cite{SchmidtSS95}, Martingales \cite{Azuma67}, matrix-valued random variables \cite{AhlswedeW02}, and many more.

We present an alternative proof of Theorem \ref{thm:Hoeffding-intro}, which differs significantly from all other proofs the authors are aware of.  Our proof is derived from recent work on algorithmic stability for adaptive data analysis~\cite{BassilyNSSSU16}.  At a high level, our proof shows that Theorem \ref{thm:Hoeffding-intro} follows from an appropriate bound on the expectation of the maximum of $\T$ independent realizations of $\sum_{i=1}^n X_i-\mu_i$. The required bound on the expected maximum follows from the existence of a \emph{stable} randomized algorithm for selecting an approximate maximizer from among the $\T$ realizations. Here stability means that the output distribution of the algorithm is insensitive to small changes in these values, which is formalized by differential privacy~\cite{DworkMNS06}.

The alternative proof was (in hindsight) a key ingredient in obtaining tight bounds on the generalization properties of differential privacy~\cite{BassilyNSSSU16}, which can be viewed as a strong extension of Theorem~\ref{thm:Hoeffding-intro}.  Thus we believe our approach both deepens our understanding of concentration inequalities and may find further applications.

\subsection{Our Approach}
Since reasoning directly about the distribution of $\sum_{i=1}^{n} X_i - \mu_i$ is difficult, proofs of Theorem~\ref{thm:Hoeffding-intro} typically work by bounding the expectation of some random variable that serves as a ``proxy'' for the tail probability $\prob{\sum_{i=1}^{n} X_i - \mu_i \geq \eps n}$.  Our proof uses a very different proxy than previous proofs. (See Section~\ref{sec:comparison} for a comparison to other proofs.)  Specifically, to bound $\prob{Y \geq y}$ we bound the quantity $\expe{\max\{0, Y^1, Y^2, \cdots, Y^\T\}}$, where $Y^1, \cdots, Y^\T$ are independent copies of the original random variable $Y$.\footnote{Our proof also works with the proxy $\expe{\max\{|Y^1|, |Y^2|, \cdots, |Y^\T|\}}$, which naturally yields two-sided tail bounds.}  Bounding this proxy implies a tail bound via the following lemma.
\begin{lem}\label{lem:MaxTB-intro}
Let $Y$ be a random variable and let $Y^1, Y^2, \cdots, Y^\T$ be independent copies of $Y$. Then $$\prob{Y \geq 2 \expe{\max\left\{0,Y^1,\cdots,Y^\T\right\}}} \leq \frac{\ln(2)}{\T}.$$
\end{lem}
\begin{proof}
Let $y= 2 \expe{\max\{0,Y^1,\cdots,Y^\T\}}$ and $\delta = \prob{Y \geq y}$. By Markov's inequality,\footnote{Markov's inequality states that, for a non-negative random variable $X$ and a constant $x>0$, $\prob{X \geq x} \leq \expe{X}/x$.} $$\prob{\max\{0, Y^1, Y^2, \cdots, Y^\T\} \geq y} \leq \frac12.$$ However, if $\delta > \ln(2)/\T$, then 
\begin{align*}
\prob{\max\{0, Y^1, Y^2, \cdots, Y^\T\}\geq y} &= 1 - \prob{\forall j \in [m]~~~Y^j < y}\\
&=1-\prob{Y < y}^\T= 1-(1-\delta)^\T \\
&> 1-e^{-\delta \T} > 1-e^{-\ln(2)} = 1/2,
\end{align*}
which is a contradiction. Thus $\delta \leq \ln(2)/\T$, as required.
\end{proof}

Thus, the key technical step in our proof is to establish the following proxy bound.
\begin{prop}\label{prop:Max-intro}
Let $X_1, \cdots, X_n$ be independent random variables supported on $[0,1]$ and $\mu_i=\expe{X_i}$ for each $i$. Define $Y = \sum_{i=1}^n X_i - \mu_i$. Then for every $\T \in \N$, if $Y^1, \cdots, Y^\T$ are independent copies of $Y$,
$$
\expe{\max\left\{ 0, Y^1, \cdots, Y^\T \right\}} \leq O(\sqrt{n \cdot \ln \T} ).
$$
\end{prop}
Combining Proposition \ref{prop:Max-intro} with Lemma \ref{lem:MaxTB-intro} and setting $
\T=\e{\Theta(\varepsilon^2 n)}$ yields Theorem \ref{thm:Hoeffding-intro}. We remark that Proposition \ref{prop:Max-intro} and Theorem \ref{thm:Hoeffding-intro} are equivalent,\footnote{More precisely, for \emph{any} random variable $X$, the condition $\forall x \geq 0 ~~ \prob{X \geq x} \leq e^{-\Omega(x^2)}$ is equivalent to the condition $\forall m ~~ \ex{}{\max\{0,X^1, \cdots, X^m\}} \leq O(\ln m)$ (as well as other equivalent characterizations). Such a random variable is called \emph{subgaussian} \cite{Rivasplata12}.}  because the proxy bound can be derived by integrating the tail bound of  Theorem \ref{thm:Hoeffding-intro} combined with a union bound.

Our proof of Proposition \ref{prop:Max-intro} is essentially a special case of the work of~\cite{BassilyNSSSU16} on algorithmic stability, which proves concentration inequalities in a more general ``adaptive'' setting.  

\subsubsection*{Bounding the Proxy via Stability}

Define $Y^{\T+1}=0$ so that $$\max\left\{ 0, Y^1, \cdots, Y^\T \right\} = \max\left\{ Y^1, \cdots, Y^\T, Y^{\T+1} \right\} = Y^{\argmax_{\t \in [\T+1]} Y^\t}.$$  The function $f(Y^1, \cdots, Y^{\T+1}) = \argmax_{\t \in [\T+1]} Y^\t$ is extremely \emph{unstable} in the sense that a small change in the inputs $Y^1, \cdots, Y^{\T+1}$ could change the output arbitrarily.  In contrast, a constant function is perfectly \emph{stable}, since it doesn't depend on its input at all.  Our proof relies on the existence of an intermediate function that approximates the $\argmax$ function but provides some stability.  We introduce a parameter $\stab > 0$ that parameterizes the degree of stability, and define a procedure $\sel_{\stab}$, and analyze the quantity $\expe{Y^{\sel_\stab(Y^1, \cdots, Y^{\T+1})}}$. 

Smaller values of $\eta$ imply a higher degree of stability.  Specifically, as $\eta \rightarrow 0$ then $\sel_{\stab}$ is independent of its input, so $\expe{Y^{\sel_\stab(Y^1, \cdots, Y^{\T+1})}} \rightarrow 0$, since $\expe{Y}=0$. On the other hand, as $\eta \to \infty$, the approximation becomes arbitrarily accurate and we have 
$$
\expe{Y^{\sel_\stab(Y^1, \cdots, Y^{\T+1})}} \to \expe{\max\left\{0, Y^1, \cdots, Y^{\T} \right\}}.
$$
All that remains is to quantify the tradeoff between stability and accuracy, and find a suitable value of $\stab > 0$ that allows us to relate $\expe{Y}$ to $\mathbb{E}[\max\{0, Y^1, \cdots, Y^{\T}\}]$.  The full proof is contained in Section~\ref{sec:Proof}.

\medskip

Subsequent to this work, Nissim and Stemmer \cite{NissimS17} used the same techniques as us to prove novel concentration inequalities which are applicable to more ``heavy-tailed'' distributions. Specifically, they extend McDiarmid's inequality (see \S\ref{sec:mcdiarmid}) to functions with high worst-case sensitivity but low ``average-case'' sensitivity.

\subsection{Comparison to Previous Proofs} \label{sec:comparison}

For context, we give some comparison between our approach and the previous proofs of concentration inequalities.  Previous proofs use different proxies for the tails $\prob{\sum_{i=1}^{n} X_i - \mu_i \geq \eps n}$.  The most common proxy is the \emph{moment generating function} $f(\lambda) := \expe{\e{\lambda\sum_{i=1}^n X_i - \mu_i}}$.  If $X_1, \cdots, X_n$ are independent random variables ssupported on $[0,1]$ and $\mu_i=\expe{X_i}$ for each $i$, then 
\begin{equation}
\forall \lambda \in \R~~~~~\expe{\e{\lambda \sum_{i=1}^n X_i - \mu_i}} \leq \e{O(\lambda^2 n)}.
\label{eqn:MGF}\end{equation}
By Markov's inequality, \eqref{eqn:MGF} implies the tail bound
\begin{equation} 
\forall \lambda \in \R~~~~~\prob{\sum_{i=1}^n X_i -\mu_i \geq \varepsilon n} 
\leq \frac{\expe{\e{\lambda\left(\sum_{i=1}^n X_i - \mu_i\right)}}}{\e{\lambda \varepsilon n}} 
\leq \e{O(\lambda^2 n) - \lambda \varepsilon n}.\label{eqn:Markov-intro}
\end{equation}
Setting $\lambda=\Theta(\varepsilon)$ appropriately in \eqref{eqn:Markov-intro} yields Theorem \ref{thm:Hoeffding-intro}. The reverse is also true---integrating the tail bound of Theorem \ref{thm:Hoeffding-intro} yields the bound on the moment generating function \eqref{eqn:MGF}. Thus Theorem \ref{thm:Hoeffding-intro} is equivalent to the proxy bound \eqref{eqn:MGF}.
The moment generating function is particularly easy to work with, as independence can be exploited to ``factor'' the moment generating function: $\expe{\e{\lambda\sum_{i=1}^n X_i - \mu_i}} = \prod_{i=1}^n \expe{\e{\lambda\left(X_i - \mu_i\right)}}.$ Thus proving \eqref{eqn:MGF} reduces to analysing a single $X_i-\mu_i$ random variable.

The \emph{moments} $m(k) := \mathbb{E}[( \sum_{i=1}^n X_i - \mu_i )^k]$ are another common proxy. If $X_1, \cdots, X_n$ are independent random variables supported on $[0,1]$ and $\mu_i=\expe{X_i}$ for each $i$, then
\begin{equation}
\forall k \in \N~~~~~\expe{\left( \sum_{i=1}^n X_i - \mu_i \right)^{2k}} \leq O\left( nk \right)^k.
\label{eqn:Moments}\end{equation}
Theorem \ref{thm:Hoeffding-intro} follows from \eqref{eqn:Moments} by applying Markov's inequality with an appropriate choice of $k = \Theta(\eps^2 n)$.
Again, the moment bound \eqref{eqn:Moments} can be obtained by integrating (a two-sided version of) the tail bound in Theorem \ref{thm:Hoeffding-intro}, so moment bounds are equivalent to tail bounds. 
An immediate benefit of using moment bounds is that the independence assumption in Theorem \ref{thm:Hoeffding-intro} can be relaxed to \emph{limited independence} \cite{SchmidtSS95,BunS15}.

Impagliazzo and Kabanets \cite{ImpagliazzoK10} use as a proxy $\expe{\prod_{i \in I} X_i}$ where $I \subseteq \{1, 2, \cdots, n\}$ is a random set of indices. Other proofs of Theorem \ref{thm:Hoeffding-intro} tend to resort to a direct analysis of the binomial distribution.  We refer the reader to \cite{MulzerChernoffNotes} for a compendium of proofs.


\subsubsection*{The Advantages of Different Proxies}

The proxy $\mathbb{E}[\max\{ 0, Y^1, \cdots, Y^\T \}]$ is a very natural quantity to bound. In many applications, tail bounds are combined with a union bound because the  real quantity of interest is of the form $\max\{ Y^1, \cdots, Y^\T \}$.\footnote{Proposition \ref{prop:Max-intro} also holds if the random variables $Y^1, \cdots, Y^\T$ are not independent, like the union bound.}

The advantage of our proxy $\mathbb{E}[\max\{ 0, Y^1, \cdots, Y^\T \}]$ is that it less sensitive to heavy tails than either moments or the moment generating function. Namely, suppose that, with probability $\delta$, the random variable $Y$ takes on a large value $\Delta$. Then $\mathbb{E}[\max\{ 0, Y^1, \cdots, Y^\T \}]$ grows linearly with $\Delta \to \infty$ as $O(m\delta\Delta)$, whereas $\E[Y^{2k}]$ grows polynomially as $\delta \Delta^{2k}$ and $\E[e^{\lambda Y}]$ grows exponentially as $O(\delta e^{\lambda \Delta})$. Note that $\prob{Y \geq y}$ does not grow at all as $\Delta$ increases beyond $y$.

For Theorem \ref{thm:Hoeffding-intro} the tails are thin enough that all three proxies yield the same result. However, in the application of Bassily et al.~\cite{BassilyNSSSU16} to adaptive data analysis, there is a $\delta$ tail event as described above. Thus, using the proxy $\mathbb{E}[\max\{ 0, Y^1, \cdots, Y^\T \}]$ yields tighter bounds than what is achievable using moments. There may be other applications for which our proxy is more suitable than either moments or the moment generating function.

We also note that our proxy can be bounded in terms of the moment generating function: By Jensen's inequality and the bound $\max\{1, e^{\lambda Y^1}, \cdots, e^{\lambda Y^\T}\} \leq 1 + e^{\lambda Y^1} + \cdots + e^{\lambda Y^\T}$ we can show that
\begin{equation}
\forall \lambda > 0~~~~~\expe{\max\left\{ 0, Y^1, \cdots, Y^\T \right\}} \leq \frac{1}{\lambda} \log \expe{e^{\lambda \cdot \max\left\{ 0, Y^1, \cdots, Y^\T \right\}}} \leq \frac{1}{\lambda} \log \left( 1 + \T \cdot \expe{e^{\lambda Y}} \right).
\end{equation}
Thus Proposition \ref{prop:Max-intro} can be derived from the moment generating function bound \eqref{eqn:MGF} by setting $\lambda=\Theta\left(\sqrt{\ln (m+1)/n}\right)$ appropriately. Likewise, Proposition \ref{prop:Max-intro} can be derived from the moment bound \eqref{eqn:Moments} via 
\begin{equation}
\forall k \in \mathbb{N} ~~~~~ \expe{\max\left\{|Y^1|, \cdots, |Y^\T|\right\}} \leq \ex{}{\max\{(Y^1)^{2k}, \cdots, (Y^\T)^{2k}\}}^{1/2k} \leq \left( m \cdot \expe{Y^{2k}} \right)^{1/2k}.
\end{equation}

\section{Proof of the Proxy Bound}\label{sec:Proof}\tnote{better name needed} \jnote{Agreed.  This one at least has words specific to our paper :)}
Let $\matX$ be a random $n \times \T$ matrix with entries supported on $[0,1]$.  Let $X_i^\t$ denote the random variable corresponding to the entry in the $i^\text{th}$ row and $\t^\text{th}$ column of $\matX$.\footnote{We choose the subscript-superscript notation $\matX_i^\t$, rather than the more common double-subscript notation $\matX_{i,\t}$, because the rows and columns of $\matX$ play markedly different roles that we want to emphasize.}  As a convention we will use lower-case letters $\matx$ and $\matx_{i}$ and $x^{\t}_{i}$ to denote realizations of the matrix, rows of the matrix, and entries of the matrix, respectively, and we will use upper case letters to denote the corresponding random variables.

\begin{lemma}[Main Lemma] \label{lem:main1}
If $\matX$ is a random $n \times \T$ matrix with entries supported on $[0,1]$ and independent rows, then
$$
\forall \stab > 0 ~~~~~ \expe{\max_{\t \in [\T]} \sum_{i=1}^{n} X^\t_i} \leq e^{\stab} \max_{\t \in [\T]} \expe{\sum_{i=1}^{n} X^\t_i} + \frac{2 \ln(\T)}{\stab}.
$$
\end{lemma}
Before delving into the proof, some remarks are in order. Lemma \ref{lem:main1} states that, under certain independence and boundedness conditions, we can ``switch'' the expectation and the maximum of a collection of random sums at the price of a multiplicative $e^{\stab}$ factor and an additive $2\ln(\T)/\stab$ factor, for any choice of $\stab > 0$. 

In our application, we will start with some independent bounded random variables $(X_1,\dots,X_n)$.  Each column $\t$ of $\matX$ will an independent copy of these random variables, $(X_1^\t,\dots,X_n^\t)$.  In this case the columns of $\matX$ are also independent. 
The lemma then says that the the expected maximum of the column sums is not too much larger than the expectation of a single column sum, which implies our tail bounds.
\begin{proof}[Proof of Lemma \ref{lem:main1}]
The proof relies on the existence of a randomized \emph{stable selection procedure} $\sel_{\stab} : [0,1]^{n \times \T} \to [\T]$ defined by $$\prob{\sel_{\stab}(\matx) = \t}  = \frac{1}{C_{\stab}(\matx)}  \exp\left( \tfrac{\stab}{2} \cdot \sum_{i=1}^{n} x^\t_i \right),~~~~~~\textrm{where}~~~~~~C_{\stab}(\matx) =  \sum_{\t =1}^{ \T}\exp\left( \tfrac{\stab}{2} \cdot \sum_{i=1}^{n} x^\t_i \right).$$ The \emph{stability parameter} $\stab>0$ can be chosen arbitrarily.
The key to the proof is that $\sel$ balances two goals: 
\begin{itemize}
\item (Stability) the distribution of $\sel_{\stab}(\matx)$ is not very sensitive to changing a single row of $\matx$ (i.e.~the probability mass function can only change by a small multiplicative factor), and
\item (Accuracy) in expectation, the chosen column $\t = \sel_{\stab}(\matx)$ is close to the largest column of $\matx$ (i.e.~$\sel_\stab(\matx) \approx \argmax_{\t \in [\T]} \sum_{i=1}^n x_i^\t$).
\end{itemize}

For $\matx \in [0,1]^{n \times \T}$, $\vecx \in [0,1]^\T$, and $i \in [n]$, define $(\matx_{-i},\vecx) \in [0,1]^{n \times \T}$ to be $\matx$ with the $i^\text{th}$ row replaced by $\vecx$. That is, $${(\matx_{-i},\vecx)}_{i'}^{\t'} = \left\{ \begin{array}{cl} \matx_{i'}^{\t'} & \text{if } i' \ne i \\ \vecx^{\t'} & \text{if } i'=i \end{array} \right.$$ for $i' \in [n]$ and $\t' \in [\T]$. 

The next two claims establish the stability and accuracy properties of $\sel_{\stab}$.  These two claims are special cases of well known results in the differential privacy literature.
\begin{claim}[Stability] \label{clm:emstab}
For every $\matx \in [0,1]^{n \times \T}$, $\vecx \in [0,1]^\T$, $i \in [n]$, and $\t \in [\T]$,
$$
e^{-\stab} \prob{\sel_{\stab}(\matx)=\t} \leq \prob{\sel_{\stab}(\matx_{-i},\vecx)=\t} \leq e^{\stab} \prob{\sel_{\stab}(\matx)=\t} .
$$
\end{claim}

\begin{proof}[Proof of Claim~\ref{clm:emstab}]
Since
$$
\left| \sum_{i'=1}^n \matx_{i'}^\t - \sum_{i'=1}^n (\matx_{-i},\vecx)_{i'}^\t \right| = \left| \matx_i^\t - \vecx^\t \right| \leq 1,
$$
for every $\t \in [\T]$, we have
$$
e^{-\tfrac{\stab}{2}} \leq
\frac{\exp\left( \tfrac{\stab}{2} \cdot \sum_{i'=1}^{n} x^\t_{i'} \right)}{\exp\left( \tfrac{\stab}{2} \cdot \sum_{i'=1}^{n} (\matx_{-i},\vecx)_{i'}^\t \right)} = \exp\left( \tfrac{\stab}{2} \cdot \left( \sum_{i'=1}^n \matx_{i'}^\t - \sum_{i'=1}^n (\matx_{-i},\vecx)_{i'}^\t  \right) \right) = \exp\left( \tfrac{\stab}{2} \cdot  \left( \matx_i^\t - \vecx^\t \right) \right)\leq e^{\tfrac{\stab}{2}}.$$
Consequently, $$C_\stab(\matx) =  \sum_{\t =1}^\T\exp\left( \tfrac{\stab}{2} \cdot \sum_{i'=1}^{n} x^\t_{i'} \right) \leq  \sum_{\t = 1}^\T e^{\tfrac{\stab}{2}} \cdot \exp\left( \tfrac{\stab}{2} \cdot \sum_{i'=1}^{n} (\matx_{-i},\vecx)^\t_{i'} \right) = e^{\tfrac{\stab}{2}} \cdot C_\stab(\matx_{-i},\vecx)$$ and, likewise, $C_\stab(\matx_{-i},\vecx) \leq e^{\tfrac{\stab}{2}} \cdot C_\stab(\matx)$.
The definition of $\sel_{\stab}$ implies
$$
e^{-\tfrac{\stab}{2}} \cdot e^{-\tfrac{\stab}{2}} \leq \frac{\prob{\sel_{\stab}(\matx_{-i}, \vecx) = \t}}{\prob{\sel_{\stab}(\matx) = \t}} = \frac{\exp\left( \tfrac{\stab}{2} \cdot \sum_{i'=1}^{n} x^\t_{i'} \right)}{\exp\left( \tfrac{\stab}{2} \cdot \sum_{i'=1}^{n} (\matx_{-i},\vecx)_{i'}^\t \right)} \cdot \frac{C_{\stab}(\matx_{-i},\vecx)}{C_{\stab}(\matx)} \leq e^{\tfrac{\stab}{2}} \cdot e^{\tfrac{\stab}{2}},
$$
as required.
\end{proof}

\begin{claim}[Accuracy] \label{clm:emacc}
For every $\matx \in [0,1]^{n \times \T}$,
$$
\ex{\sel_{\stab}}{\sum_{i=1}^{n} x_{i}^{ \sel_{\stab}(\matx)} } \geq \max_{\t \in [\T]} \sum_{i=1}^{n} x^{\t}_{i} - \frac{2 \ln(\T)}{\stab}
$$
\end{claim}
Observe that as $\eta \to \infty$, the accuracy becomes arbitrarily good, whereas the stability degrades.

\begin{proof}[Proof of Claim~\ref{clm:emacc}]
Recall that, by construction, $\sel_{\stab}$ satisfies, $$\prob{\sel_{\stab}(\matx) = \t}  = \frac{1}{C_{\stab}(\matx)}  \exp\left( \tfrac{\stab}{2} \cdot \sum_{i=1}^{n} x^\t_i \right),~~~~~~\textrm{where}~~~~~~C_{\stab}(\matx) =  \sum_{\t = 1}^{ \T}\exp\left( \tfrac{\stab}{2} \cdot \sum_{i=1}^{n} x^\t_i \right).$$
Rearranging, we have
$$
\sum_{i=1}^{n} x^\t_i = \frac{2}{\stab} \left(\ln C_{\stab}(\matx) + \ln \prob{\sel_{\stab}(\matx) = \t} \right).
$$
Thus,
\begin{align*}
\ex{\sel_{\stab}}{\sum_{i=1}^{n} x^{\sel_{\stab}(\matx)}_{i}} 
={} \sum_{\t = 1}^{\T} \prob{\sel_{\stab}(\matx) = \t} \cdot \sum_{i=1}^{n} x^{\t}_{i}
={} &\sum_{\t = 1}^{\T} \prob{\sel_{\stab}(\matx) = \t} \cdot \frac{2}{\stab} \left(\ln C_{\stab}(\matx) + \ln \prob{\sel_{\stab}(\matx) = \t} \right) \\
={} &\frac{2}{\stab}\left( 1 \cdot \ln C_{\stab}(\matx)+ \sum_{\t = 1}^{\T} \prob{\sel_{\stab}(\matx) = \t} \ln \prob{\sel_{\stab}(\matx) = \t}\right)\\
={} &\frac{2}{\stab}\left( \ln C_{\stab}(\matx) - \mathsf{H}\left[\sel_{\stab}(\matx)\right] \right),
\end{align*}
where $\mathsf{H}\left[\sel_{\stab}(\matx)\right] $ is the Shannon entropy of the random variable $\sel_{\stab}(\matx)$ measured in ``nats,'' rather than bits, as we are using natural logarithms, rather than binary logarithms. Since the random variable $\sel_{\stab}(\matx)$ is supported on a set of size $\T$, the entropy satisfies $\mathsf{H}\left[\sel_{\stab}(\matx)\right] \leq \ln(\T)$, as the entropy function is maximized by the uniform distribution \cite[Theorem 2.6.4]{cover2012elements}.  By definition we have $$C_{\stab}(\matx)  =  \sum_{\t = 1}^{ \T}\exp\left( \tfrac{\stab}{2} \cdot \sum_{i=1}^{n} x^\t_i \right) \geq \max_{\t \in [\T]} \exp\left({\tfrac{\stab}{2} \sum_{i=1}^{n} x^{\t}_{i}}\right)~~~\text{and}~~~\ln C_{\stab}(\matx) \geq \max_{\t \in [\T]} \frac{\stab}{2} \cdot \sum_{i=1}^{n} x^{\t}_{i}.$$
The claim now follows from these two inequalities: $$\ex{\sel_{\stab}}{\sum_{i=1}^{n} x^{\sel_{\stab}(\matx)}_{i}} = \frac{2}{\stab}\left( \ln C_{\stab}(\matx) - \mathsf{H}\left[\sel_{\stab}(\matx)\right] \right) \geq \frac{2}{\stab}\left(\max_{\t \in [\T]}  \frac{\stab}{2} \cdot \sum_{i=1}^{n} x^{\t}_{i} - \ln(\T) \right).$$
\end{proof}
In order to complete the proof of the lemma, we have the following claim. This claim uses stability (Claim \ref{clm:emstab}) to show that $\ex{\matX,\sel_{\stab}}{\sum_{i=1}^{n} X_{i}^{ \sel_{\stab}(\matX)} }$ is not much larger than $\ex{\matX,\sel_{\stab}}{\sum_{i=1}^{n} X_{i}^{\t} }$ for a fixed $\t \in [\T]$.
\begin{claim}\label{clm:hybrid} If $\expe{\sum_{i=1}^n X_i^\t} \leq \mu$ for all $\t \in [\T]$, then
$$
\ex{\matX,\sel_{\stab}}{\sum_{i=1}^{n} X_{i}^{ \sel_{\stab}(\matX)} } \leq e^{\stab} \mu.
$$
\end{claim}
\begin{proof}[Proof of Claim~\ref{clm:hybrid}]
Let $\tilde\matX$ be an independent and identical copy of $\matX$. Let $\tilde\matX_i \in [0,1]^\T$ denote the $i^\text{th}$ row of $\tilde\matX$. The key observation is that the pair $(\matX_{-i},\tilde\matX_i)$ and $X_i^\t$ is identically distributed to the pair $\matX$ and $\tilde X_i^\t$ --- this is where we use the independence assumption. Thus the pair $\sel_\stab(\matX_{-i},\tilde\matX_i)$ and $X_i^\t$ is identically distributed to the pair $\sel_\stab(\matX)$ and $\tilde{X}_i$, which means we can swap them below. By Claim \ref{clm:emstab},
\begin{align*}
\ex{\matX,\sel_{\stab}}{\sum_{i=1}^{n} X_{i}^{ \sel_{\stab}(\matX)} } =& \ex{\matX}{\sum_{\t =1}^\T \sum_{i=1}^n \pr{\sel_{\stab}}{\sel_{\stab}(\matX)=\t} X_i^\t}\\
\leq& \ex{\matX,\tilde\matX}{\sum_{\t =1}^\T \sum_{i=1}^n e^{\stab}\pr{\sel_{\stab}}{\sel_{\stab}(\matX_{-i},\tilde\matX_i)=\t} X_i^\t}\\
=& \ex{\matX,\tilde\matX}{\sum_{\t =1}^\T \sum_{i=1}^n e^{\stab}\pr{\sel_{\stab}}{\sel_{\stab}(\matX)=\t} \tilde X_i^\t}\\
\leq& \ex{\matX,\tilde\matX}{\sum_{\t =1}^\T e^{\stab}\pr{\sel_{\stab}}{\sel_{\stab}(\matX)=\t} \mu}\\
=& e^{\stab} \mu.
\end{align*}
\end{proof}
Combining Claim \ref{clm:emacc} with Claim \ref{clm:hybrid} (setting $\mu = \max_{\t \in [\T]} \expe{\sum_{i=1}^n \matX_i^\t}$), we have $$\ex{\matX}{\max_{\t \in [\T]} \sum_{i=1}^{n} X^{\t}_{i} - \frac{2 \ln(\T)}{\stab}} \leq \ex{\matX,\sel_{\stab}}{\sum_{i=1}^{n} X_{i}^{ \sel_{\stab}(\matx)} } \leq e^{\stab} \mu = e^\stab \max_{\t \in [\T]} \ex{\matX}{\sum_{i=1}^{n} X_{i}^{\t} }.$$ Rearranging yields the lemma.
\end{proof}

Lemma \ref{lem:main1} readily yields the bound claimed in the introduction:

\begin{prop}[Proposition \ref{prop:Max-intro}]\label{prop:Max-later}Let $X_1, \cdots, X_n$ be independent random variables supported on $[0,1]$ and $\mu_i=\expe{X_i}$ for each $i$. Define $Y = \sum_{i=1}^n X_i - \mu_i$. Fix $\T \in \mathbb{N}$ and let $Y^1, \cdots, Y^\T$ be independent copies of $Y$. Then $$\expe{\max\left\{ 0, Y^1, \cdots, Y^\T \right\}} \leq 4 \sqrt{n \cdot \ln (\T+1)} .$$\end{prop}
\begin{proof} 
Firstly, if $\T \geq e^n-1$, then the result holds trivially as $\max\left\{ 0, Y^1, \cdots, Y^\T \right\} \leq n$ with certainty. So we may assume $\T < e^n-1$.

Let $\mu = \sum_{i=1}^n \mu_i$. 
For each $i \in [n]$, let $X^1_i, \cdots, X^\T_i$ be independent copies of $X_i$, so that $Y^\t = \sum_{i=1}^n X_i^\t - \mu_i$ for all $\t \in [\T]$. Let $X_i^{\T+1} = \mu_i$ be a constant ``dummy random varaible'' for each $i$.

Now we apply Lemma \ref{lem:main1} to the random matrix $\matX \in [0,1]^{n \times (\T+1)}$: $$
\forall \stab > 0 ~~~~~ \expe{\max_{\t \in [\T+1]} \sum_{i=1}^{n} X^\t_i} \leq e^{\stab} \max_{\t \in [\T+1]} \expe{\sum_{i=1}^{n} X^\t_i} + \frac{2 \ln(\T+1)}{\stab}.
$$

By construction, $\expe{\sum_{i=1}^{n} X^\t_i} = \sum_{i=1}^n \mu_i = \mu$ for all $\t \in [\T+1]$. Also $\sum_{i=1}^n X_i^\t = Y^\t + \mu$ for all $\t \in [\T]$ and $\sum_{i=1}^n X_i^{\T+1} = 0 + \mu$. Substituting in these expressions yields
$$\forall \stab > 0 ~~~~~ \expe{\max\left\{Y^1 + \mu, Y^2 + \mu, \cdots, Y^\T+\mu, 0+\mu \right\}} \leq e^{\stab} \mu + \frac{2 \ln(\T+1)}{\stab}.$$
Subtracting $\mu$ gives $$\forall \stab > 0 ~~~~~ \expe{\max\left\{Y^1, Y^2, \cdots, Y^\T, 0 \right\}} \leq \left( e^{\stab} -1 \right) \mu + \frac{2 \ln(\T+1)}{\stab}.$$
Finally, we use the (crude) bound $\mu \leq n$ and the approximation $e^\stab -1 \leq 2\stab$ for $\stab < 1$ to obtain $$\forall \stab \in (0,1) ~~~~~ \expe{\max\left\{Y^1, Y^2, \cdots, Y^\T, 0 \right\}} \leq 2\stab n + \frac{2 \ln(\T+1)}{\stab}.$$
Set $\stab = \sqrt{\ln(\T+1)/n} < \sqrt{\ln(e^n)/n} = 1$ to complete the proof.
\end{proof}


Now we can prove the subgaussian tail bound using the proxy bound of Proposition \ref{prop:Max-later}.

\begin{theorem}[Theorem \ref{thm:Hoeffding-intro}]\label{thm:Hoeffding-later}
If $X_1, \cdots, X_n$ are independent random variables supported on $[0,1]$ and $\mu_i=\expe{X_i}$ for every $i$, then $$\forall \varepsilon \geq 0 ~~~~~ \prob{\sum_{i=1}^n X_i -\mu_i \geq \varepsilon n} \leq e^{1-\varepsilon^2 n / 64}.$$

\tnote{Hmm. We are off by a factor of 128.} \jnote{Yeah, bummer.  There are several inefficiencies.}
\end{theorem}
Needless to say, we have not optimized the constants in Theorem \ref{thm:Hoeffding-later}.
\begin{proof}
Fix $\T \in \mathbb{N}$ to be determined later.
Let $Y = \sum_{i=1}^{n} X_i - \mu$ and let $Y^1, \cdots, Y^\T$ be independent copies of $Y$.
By Proposition \ref{prop:Max-later}, we have $$\expe{\max\{0,Y^1,\cdots,Y^\T\}} \leq 4 \sqrt{n \ln (\T + 1)}.$$
By Lemma \ref{lem:MaxTB-intro}, $\prob{Y \geq 2\expe{\max\{0,Y^1,\cdots,Y^\T\}}} \leq \ln(2)/\T.$  Now, set $m=\left\lfloor e^{\varepsilon^2 n / 64} -1 \right\rfloor$, so that $8\sqrt{n \ln(m+1)} \leq \varepsilon n$.
Then 
$$
\prob{Y \geq \varepsilon n} \leq \prob{Y \geq 8 \sqrt{n \ln(m+1)}} \leq \prob{Y \geq 2\expe{\max\{0,Y^1,\cdots,Y^\T\}}} \leq \frac{\ln(2)}{\T} \leq \frac{\ln(2)}{e^{\varepsilon^2 n / 64} -2}.
$$
Thus $$\prob{Y \geq \varepsilon n} \leq \min \left\{ 1, \frac{\ln(2)}{e^{\varepsilon^2 n / 64} -2} \right\} \leq (2+\ln 2) \cdot e^{-\varepsilon^2 n / 64} \leq e^{1-\varepsilon^2 n / 64}.$$
\end{proof}

\section{Extensions}

To demonstrate the versatility of our techniques, we show how they can be used to prove some well-known generalizations of Theorem \ref{thm:Hoeffding-intro}. 

\subsection{Multiplicative Tail Bound}
\begin{theorem}
If $X_{1},\dots,X_{n}$ are independent random variables supported on $[0,1]$ and $\mu=\expe{\sum_{i=1}^{n} X_i} $, then
$$
\forall \varepsilon \geq 0 ~~~~~~
\prob{\sum_{i=1}^{n} X_i \geq (1+\varepsilon) \mu} \leq e \cdot \left(1 + \frac{\varepsilon}{4}\right)^{-\varepsilon \mu/8}.
$$
Moreover, the above bound implies
$$
\forall \eps \in [0,10]~~~~~\prob{\sum_{i=1}^{n} X_i \geq (1+\varepsilon) \mu} \leq e^{1-\varepsilon^2 \mu/64}.
$$
\end{theorem}
\begin{proof}
Let $\delta = \prob{\sum_{i=1}^{n} X_i \geq (1+\varepsilon) \mu}$ and $\T=\lceil \ln(2)/\delta \rceil$.
For each $i \in [n]$, let $X^1_i, \cdots, X^\T_i$ be independent copies of $X_i$ and let $X^{\T+1}_i=\expe{X}$ be a constant. Now we apply Lemma \ref{lem:main1} to the random matrix $\matX \in [0,1]^{n \times (\T+1)}$ to obtain
\begin{equation} \label{eq2}
\expe{\max_{\t \in [\T+1]} \sum_{i=1}^{n} X^\t_i } \leq e^{\stab} \mu + \frac{2 \ln(\T+1)}{\stab}.
\end{equation}
On the other hand, $$\expe{\max_{\t \in [\T+1]} \sum_{i=1}^{n}\! X^\t_i } \geq \mu + \varepsilon \mu \cdot \prob{\max_{\t \in [\T]} \sum_{i=1}^{n} \! X^\t_i \geq (1\!+\!\varepsilon) \mu} \!=\! \mu + \varepsilon \mu (1-(1-\delta)^\T) \geq \mu + \varepsilon \mu (1-e^{-\delta\T}) \geq \mu \left( 1 \!+\! \frac{\varepsilon}{2} \right).$$
Combining this with \eqref{eq2} and subtracting $\mu$ from both sides gives \begin{equation*} \forall \stab>0 ~~~~~~~~~~\frac{\varepsilon}{2} \mu \leq \left( e^\stab - 1 \right) \mu + \frac{2 \ln (\T+1)}{\stab}.\label{eqn:multbd}\end{equation*}
Now we set $\stab = \ln(1+\varepsilon/4)$ and rearrange to get
$ \ln(\T+1) \geq \frac{\varepsilon}{8} \mu \ln\left(1+\frac{\varepsilon}{4}\right).$
Since $m \leq \ln(2)/\delta + 1$, we have $$\prob{\sum_{i=1}^{n} X_i \geq (1+\varepsilon) \mu} = \delta \leq \min\left\{ 1, \frac{\ln 2}{m-1} \right\} \leq \frac{2+\ln 2}{m+1} \leq e \left(1 + \frac{\varepsilon}{4}\right)^{-\varepsilon \mu/8}.$$
This gives the first half of the theorem. The second half follows from the fact that, if $0 \leq \varepsilon \leq 10$, then $1+\varepsilon/4 \geq \exp(\varepsilon/8)$.
\end{proof}

We can also bound $\prob{\sum_{i=1}^{n} X_i \leq (1-\varepsilon) \mu}$. This requires changing Lemma \ref{lem:main1} to the following.
\begin{lemma} \label{lem:main-neg}
If $\matX$ is a random $n \times \T$ matrix with entries supported on $[0,1]$ and independent rows, then
$$
\forall \stab > 0 ~~~~~ \expe{\min_{\t \in [\T]} \sum_{i=1}^{n} X^\t_i} \geq e^{-\stab} \min_{\t \in [\T]} \expe{\sum_{i=1}^{n} X^\t_i} - \frac{2 \ln(\T)}{\stab}.
$$
\end{lemma}
The proof of Lemma \ref{lem:main-neg} is almost identical to that of Lemma \ref{lem:main1}, except the maximum is replaced with the minimum, $\stab$ is replaced with $-\stab$, and some inequalities are reversed.

Lemma \ref{lem:main-neg} yields the following bound.
\begin{theorem}
If $X_{1},\dots,X_{n}$ are independent random variables supported on $[0,1]$ and $\mu=\expe{\sum_{i=1}^{n} X_i} $, then
$$
\forall \varepsilon \in [0,1] ~~~~~
\prob{\sum_{i=1}^{n} X_i \leq (1-\varepsilon) \mu} \leq e^{1-\varepsilon^2 \mu/32}.
$$
\end{theorem}

\subsection{McDiarmid's Inequality}\label{sec:mcdiarmid}

Subgaussian tail bounds are not specific to summation. For independent random variables $X_1, \cdots, X_n$, we can prove subgaussian tail bounds on $f(X_1, \cdots, X_n)$ for any ``low-sensitivity'' function $f$. The sum --- that is, $f(x_1, \cdots, x_n) = \sum_{i=1}^n x_i$ --- is only one example of a low-sensitivity function. The more general property that we require is below.

\begin{definition}
A function $f : \mathcal{X}^n \to \mathbb{R}$ is sensitivity-$\Delta$ if, for all $i \in [n]$ and all $x_1, x_2, \cdots, x_n, x'_i \in \mathcal{X}$, we have $$\left| f(x_1, x_2, \cdots, x_n) - f(x_1, x_2, \cdots, x_{i-1}, x'_i, x_{i+1}, \cdots, x_n) \right| \leq \Delta.$$
\end{definition}

We can generalize Proposition \ref{prop:Max-intro} to low-sensitivity functions as follows.

\begin{lemma}\label{lem:sens}
Let $\matX$ be a random $n \times \T$ matrix with the entries supported on $\mathcal{X}$ (not necessarily being real numbers). Assume that the rows of $\matX$ are independent. Let $f : \mathcal{X}^n \to \mathbb{R}$ be sensitivity-$\Delta$. Then
$$
\expe{\max_{\t \in [\T]} \left( f(X_1^\t, \cdots, X_n^\t) - \expe{ f(X_1^\t, \cdots, X_n^\t)} \right)} \leq 8\Delta\sqrt{n \ln m}.
$$
\end{lemma}
The proof of Lemma \ref{lem:sens} is similar to that of Theorem 7.2 of Bassily et al.~\cite{BassilyNSSSU16}.  Lemma \ref{lem:sens} allows us to generalize Theorem \ref{thm:Hoeffding-intro} to McDiarmid's inequality:

\begin{theorem}[\cite{mcdiarmid1989method}]
Let $X_1, \cdots, X_n$ be independent random variables and let $f : \mathcal{X}^n \to \mathbb{R}$ have sensitivity-$\Delta$. Then, for all $\varepsilon \geq 0$, $$\prob{f(X_1, \cdots, X_n) - \expe{f(X_1, \cdots, X_n)} \geq \varepsilon n \Delta} \leq e^{-\Omega(\varepsilon^2 n)}.$$
\end{theorem}

\bibliographystyle{alpha}
\bibliography{refs}

\newcommand{\etalchar}[1]{$^{#1}$}
\begin{thebibliography}{DMNS06}

\bibitem[AW02]{AhlswedeW02}
Rudolf Ahlswede and Andreas Winter.
\newblock Strong converse for identification via quantum channels.
\newblock {\em IEEE Transactions on Information Theory}, 48(3):569--579, 2002.

\bibitem[Azu67]{Azuma67}
Kazuoki Azuma.
\newblock Weighted sums of certain dependent random variables.
\newblock {\em Tohoku Mathematical Journal, Second Series}, 19(3):357--367,
  1967.

\bibitem[Ber24]{Bernstein24}
Sergei Bernstein.
\newblock On a modification of chebyshev?s inequality and of the error formula
  of laplace.
\newblock {\em Ann. Sci. Inst. Sav. Ukraine, Sect. Math}, 1(4):38--49, 1924.

\bibitem[BNS{\etalchar{+}}16]{BassilyNSSSU16}
Raef Bassily, Kobbi Nissim, Adam Smith, Thomas Steinke, Uri Stemmer, and
  Jonathan Ullman.
\newblock Algorithmic stability for adaptive data analysis.
\newblock In {\em Proceedings of the 48th Annual ACM SIGACT Symposium on Theory
  of Computing}, pages 1046--1059. ACM, 2016.

\bibitem[BS15]{BunS15}
Mark Bun and Thomas Steinke.
\newblock Weighted polynomial approximations: Limits for learning and
  pseudorandomness.
\newblock In {\em RANDOM}, volume~19, pages 625--644. Leibniz International
  Proceedings in Informatics, 2015.

\bibitem[Che52]{Chernoff52}
Herman Chernoff.
\newblock A measure of asymptotic efficiency for tests of a hypothesis based on
  the sum of observations.
\newblock {\em The Annals of Mathematical Statistics}, pages 493--507, 1952.

\bibitem[CT12]{cover2012elements}
Thomas~M Cover and Joy~A Thomas.
\newblock {\em Elements of information theory}.
\newblock John Wiley \& Sons, 2012.

\bibitem[DMNS06]{DworkMNS06}
Cynthia Dwork, Frank McSherry, Kobbi Nissim, and Adam Smith.
\newblock Calibrating noise to sensitivity in private data analysis.
\newblock In {\em {TCC}}, 2006.

\bibitem[Hoe63]{Hoeffding63}
Wassily Hoeffding.
\newblock Probability inequalities for sums of bounded random variables.
\newblock {\em Journal of the American statistical association},
  58(301):13--30, 1963.

\bibitem[IK10]{ImpagliazzoK10}
Russell Impagliazzo and Valentine Kabanets.
\newblock Constructive proofs of concentration bounds.
\newblock In {\em Approximation, Randomization, and Combinatorial Optimization.
  Algorithms and Techniques}, pages 617--631. Springer, 2010.

\bibitem[McD89]{mcdiarmid1989method}
Colin McDiarmid.
\newblock On the method of bounded differences.
\newblock {\em Surveys in combinatorics}, 141(1):148--188, 1989.

\bibitem[Mul]{MulzerChernoffNotes}
Wolfgang Mulzer.
\newblock Chernoff bounds.
\newblock \url{https://page.mi.fu-berlin.de/mulzer/notes/misc/chernoff.pdf}.
\newblock Accessed 16-Dec-2016.

\bibitem[NS17]{NissimS17}
Kobbi Nissim and Uri Stemmer.
\newblock Concentration bounds for high sensitivity functions through
  differential privacy.
\newblock {\em arXiv preprint arXiv:1703.01970}, 2017.

\bibitem[Riv12]{Rivasplata12}
Omar Rivasplata.
\newblock Subgaussian random variables: An expository note.
\newblock \url{https://sites.ualberta.ca/~omarr/publications/subgaussians.pdf},
  2012.

\bibitem[SSS95]{SchmidtSS95}
Jeanette~P Schmidt, Alan Siegel, and Aravind Srinivasan.
\newblock Chernoff-hoeffding bounds for applications with limited independence.
\newblock {\em SIAM Journal on Discrete Mathematics}, 8(2):223--250, 1995.

\end{thebibliography}

\end{document}